\documentclass[a4paper, 10 pt]{article}

\newif\ifproofs

\proofstrue

\usepackage{geometry}[3cm]
\usepackage[numbers,sort&compress]{natbib}
\usepackage{amsmath} 
\usepackage{cleveref}
\usepackage{bm}
\usepackage{textcomp}

\usepackage{enumitem}

\usepackage{graphics} 
\usepackage{amssymb}  

\usepackage{amsthm}

\newcommand{\norm}[1]{\left\Vert #1\right\Vert}

\newcommand{\eqDef}{:=}

\newtheorem{theorem}{Theorem}
\newtheorem{corollary}[theorem]{Corrolary}
\newtheorem{proposition}[theorem]{Proposition}
\newtheorem{lemma}[theorem]{Lemma}
\newtheorem{remark}{Remark}
\newtheorem{assumption}{Assumption}
\newtheorem{example}{Example}

\newtheorem{definition}{Definition}

\Crefname{corollary}{Cor.}{Cors.}
\Crefname{equation}{Eq.}{Eqs.}
\Crefname{figure}{Fig.}{Figs.}
\Crefname{tabular}{Tab.}{Tabs.}
\Crefname{table}{Tab.}{Tabs.}
\Crefname{theorem}{Thm.}{Thms.}
\Crefname{definition}{Def.}{Defs.}
\Crefname{section}{Sec.}{Secs.}
\Crefname{proposition}{Prop.}{Props.}
\Crefname{assumption}{Asm.}{Asms.}
\Crefname{example}{Ex.}{Exs.}

\usepackage[colorinlistoftodos,bordercolor=orange,backgroundcolor=orange!20,linecolor=orange,textsize=scriptsize]{todonotes}
\usepackage{multirow}
\usepackage{pbox}
\usepackage{amsfonts}

\newcommand{\txt}{\textstyle}

\newcommand{\x}{x}
\newcommand{\xx}{\bm{x}}

\newcommand{\yy}{\bm{y}}

\newcommand{\xag}{X}
\newcommand{\xxag}{\bm{\xag}}
\newcommand{\yag}{Y}
\newcommand{\yyag}{\bm{\yag}}

\newcommand{\Sxag}{\widetilde{\X}} 

\newcommand{\rit}{\mathbb{R}}
\newcommand{\nit}{\mathbb{N}}

\newcommand{\X}{\mathcal{X}}
\newcommand{\T}{\mathcal{T}}
\newcommand{\I}{\mathcal{I}}
\newcommand{\G}{\mathcal{G}}
\newcommand{\B}{\mathcal{B}}

\newcommand{\NE}{\mathrm{NE}}

\newcommand{\Gna}{G} 

\newcommand{\cc}{\bm{c}} 

\newcommand{\diamX}{M} 

\newcommand{\F}{F} 

\newcommand{\Bcut}{\Gamma}
\newcommand{\stmon}{c_0} 
\newcommand{\Bc}{B_{\cc}} 
\renewcommand{\th}{\theta} 
\renewcommand{\t}{t}
\renewcommand{\i}{i}
\newcommand{\iti}{_{\i,\t}}
\newcommand{\thti}{_{\th,\t}}
\def\ub{\underline{b}}
\def\ob{\overline{b}}

\newcommand{\hx}{\hat{\x}}  
\newcommand{\hxx}{\hat{\xx}}
\newcommand{\hxag}{\hat{\xag}}
\newcommand{\hxxag}{\hat{\xxag}}

\newcommand{\eqd}{:=}

\newcommand{\dth}{\mathrm{d}\th}
\title{\LARGE \bf
Routing Game on Parallel Networks: \\ the Convergence of Atomic to Nonatomic
}

\author{{Paulin Jacquot\thanks{Paulin Jacquot is with  EDF Lab Saclay, Inria and CMAP, Ecole polytechnique, CNRS. {\tt\small paulin.jacquot@polytechnique.edu} } and Cheng Wan \thanks{Cheng Wan is with LMO, Universit\'e Paris-Sud; Inria Paris and RIIS, SUFE. {\tt\small cheng.wan.2005@polytechnique.org } {This work was supported in part by the PGMO foundation.}        }
}
}

\begin{document}

\maketitle

\begin{abstract}
We consider an instance of a nonatomic routing game. We assume that the network is parallel, that is, constituted of only two nodes, an origin  and a destination. We consider infinitesimal players that have a symmetric network cost, but are heterogeneous through their set of feasible strategies and  their individual utilities. We show that if an atomic routing game instance is correctly defined to approximate the nonatomic instance, then an atomic Nash Equilibrium will approximate  the nonatomic Wardrop Equilibrium. We give explicit bounds on the distance between the equilibria according to the parameters of the atomic instance. This approximation gives a method to compute the Wardrop equilibrium at an arbitrary precision.
\end{abstract}

\section{Introduction}

\textbf{Motivation.} Network routing games were first considered by Rosenthal \cite{rosenthal1973network} in their ``atomic unsplittable'' version, where a finite set of players share a network subject to congestion. Routing games found  later on many practical applications not only in transport \cite{wardrop1952some,marcotte1997equilibria}, but also in communications \cite{orda1993competitive}, distributed computing \cite{altman2002nash} or  energy \cite{atzeni2013demand}. The different models studied are of three main categories: nonatomic games (where there is a continuum of infinitesimal players), atomic unsplittable games (with a finite number of players, each one choosing a path to her destination), and atomic splittable games (where there is a finite number of players, each one choosing how to split her weight on the set of available paths).

The concept of equilibrium is central in game theory, for it corresponds to a ``stable'' situation, where no player has interest to deviate. With a finite number of players---an \emph{atomic unsplittable} game---it is captured by the concept of Nash Equilibrium \cite{nash1950equilibrium}. With an infinite number of infinitesimal  players---the \emph{nonatomic} case---the problem is different: deviations from a finite number of players have no impact, which led Wardrop to its definition of equilibria for nonatomic games \cite{wardrop1952some}.
 A typical illustration of the fundamental difference between the nonatomic and atomic splittable routing games is the existence of an exact potential function in the former case, as opposed to the latter \cite{nisan2007algorithmic}. However, when one considers the limit game of an atomic splittable game where players become infinitely many, one obtains a nonatomic instance with infinitesimal players, and expects a relationship between the atomic splittable Nash equilibria and the Wardrop equilibrium of the limit nonatomic game. This is the question we address in this paper.

\textbf{Main results.} We propose a quantitative analysis of the link between a nonatomic routing game and a family of related atomic splittable routing games, in which the number of players grows. A novelty from the existing literature is that, for nonatomic instances, we consider a very general setting where players in the continuum $[0,1]$ have specific convex strategy-sets, the profile of which being given as a mapping from $[0,1]$ to $\rit^T$. In addition to the conventional network (congestion) cost, we consider individual utility function which is also heterogeneous among the continuum of players. 
For a nonatomic game of this form, we formulate the notion of an \emph{atomic splittable approximating sequence}, composed of instances of atomic splittable games closer and closer to the nonatomic instance. 
Our main results state the convergence of Nash equilibria (NE) associated to an approximating sequence to the Wardrop equilibrium of the nonatomic instance. In particular, \Cref{thm:converge_without_u} gives the convergence of aggregate NE flows to the aggregate WE flow in $\rit^T$ in the case of convex and strictly increasing price (or congestion cost) functions without individual utility; \Cref{thm:converge_with_u} states the convergence of NE to the Wardrop equilibrium in $((\rit^T)^{[0,1]}, \norm{.}_2)$ in the case of player-specific strongly concave utility functions. For each result we provide an upper bound on the convergence rate, given from the atomic splittable instances parameters. 
An implication of these new results concerns the computation of an equilibrium of a nonatomic instance. Although computing an NE is a hard problem in general \cite{koutsoupias1999worst},   there exists several algorithms to compute an NE through its formulation  with finite-dimensional variational inequalities \cite{facchinei2007finite}. For a Wardrop Equilibrium, a similar formulation with infinite-dimensional variational inequalities can be written, but finding a solution is much harder.

\textbf{Related work.} Some results have already been given to quantify  the relation between Nash and Wardrop equilibria. Haurie and Marcotte \cite{haurie1985relationship} show that in a sequence of atomic splittable games where atomic splittable players replace themselves smaller and smaller equal-size players with constant total weight, the Nash equilibria converge to the Wardrop equilibrium of a nonatomic game. Their proof is based on the convergence of variational inequalities corresponding to the sequence of Nash equilibria, a technique similar to the one used in this paper. Wan \cite{wan2012coalition} generalizes this result to composite games where nonatomic players and atomic splittable players coexist, by allowing the atomic players to replace themselves by players with heterogeneous sizes. 

In \cite{gentile2017nash}, the authors consider an aggregative game with linear coupling constraints (generalized Nash Equilibria) and show that the Nash Variational equilibrium can be approximated with the Wardrop Variational equilibrium. However, they consider a Wardrop-type equilibrium for a finite number of players: an atomic player considers that her action has no impact on the aggregated profile. 
 They do  not study the relation between atomic and nonatomic equilibria, as done in this paper. Finally, Milchtaich \cite{milchtaich2000generic} studies atomic unsplittable and nonatomic crowding games, where players are of equal weight and each player's payoff depends on her own action and on the number of players choosing the same action. He shows that, if each atomic unsplittable player in an $n$-person finite game is replaced by $m$ identical replicas with constant total weight, the equilibria generically converge to the unique equilibrium of the corresponding nonatomic game as $m$ goes to infinity.
 Last, Marcotte and Zhu \cite{marcotte1997equilibria} consider nonatomic players with continuous types (leading to a characterization of the Wardrop equilibrium as a infinite-dimensional variational inequality) and studied the equilibrium in an aggregative game with an infinity of nonatomic players, differentiated through a linear parameter in their cost function and their feasibility sets assumed to be convex polyhedra.

\textbf{Structure.} The remaining of the paper is organized as follows:  
 in \Cref{sec:atomicNonatomicDefs}, we give the definitions of atomic splittable and nonatomic routing games. We recall the associated concepts of Nash and Wardrop equilibria, their characterization via variational inequalities, and sufficient conditions of existence. Then, in \Cref{sec:approx_games}, we give the definition of an approximating sequence of a nonatomic game, and we give our two main theorems on the convergence of the sequence of Nash equilibria to a Wardrop equilibrium of the nonatomic game. Last, in \Cref{sec:numericalExp} we provide a numerical example of an approximation of a particular nonatomic routing game.
 
 \vspace{0.5cm}
\noindent\textbf{Notation.} We use a bold font to denote vectors (e.g. $\xx$) as opposed to scalars (e.g. $x$).

\section{Splittable Routing: Atomic and Nonatomic} 
\label{sec:atomicNonatomicDefs}

\subsection{Atomic Splittable Routing Game}

An atomic splittable routing game on parallel arcs is defined with a network constituted of a finite number of parallel links (\textit{cf} \Cref{fig:network}) on which players can load some weight. Each ``link'' can be thought as a road, a communication channel or a time slot on which each user can put a load or a task. Associated to each link is a cost or ``latency'' function that depends only of the total load put on this link.
\begin{figure}[ht!]
\centering
\vspace{-0.3cm}
\begin{tikzpicture}[scale=0.35]
\node[draw,circle,scale=1] (a1)at(-5,0) {$O$};
\node[draw,circle,scale=1] (b1)at(5,0) {$D$};
\draw[->,>=latex] (a1) to[bend left=40] (b1);
\node[scale=1] (1) at (0,2.6) {$t=1, \ c_1$};
\node[scale=1] (1) at (0,1.25) {$t=2, \ c_2$};
\draw[->,>=latex] (a1) to[bend left=13] (b1);
\node[scale=1] (1) at (0,0.3) {$\cdots$};
\draw[->,>=latex] (a1) to[bend right=30] (b1) ;
\draw[->,>=latex] (a1) to[bend right=10] (b1);
\draw[->,>=latex] (a1) to (b1);
\node[scale=1] (1) at (0,-2.4) {$t=T, \ c_T$};
\end{tikzpicture}
\vspace{-0.3cm}
\caption{A parallel network with $T$ links.}
\vspace{-0.3cm}
\label{fig:network}
\end{figure}
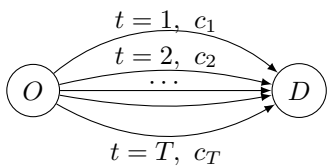
\begin{definition} \textbf{Atomic Splittable Routing Game} \\ \label{def:atomicGame}
An instance $\G$ of an atomic splittable routing game is defined by:
\begin{itemize}[leftmargin=*,wide,labelindent=-1pt]
\item a finite set of players $\I=\{1,\dots ,I\}$,
\item a finite set of arcs $\T=\{1,\dots,T\}$,
\item for each $\i\in\I$, a feasibility set $\X_i \subset \rit_+^T$,
\item for each $\i\in\I$, a utility function $u_i: \X_i \rightarrow \rit$,
\item for each $\t\in\T$, a cost or latency function $c_t(.): \rit \rightarrow \rit$ .
\end{itemize}

Each atomic player $\i\in\I$ chooses a profile $(\xx\iti)_{\t\in\T}$ in her feasible set $\X_\i$ and minimizes her cost function:
\hspace{-0.3cm}\begin{align}
\label{eq:cost_player_def}
& f_\i(\xx_\i,\xx_{-\i}):=\sum_{\t\in\T}\x\iti c_t\Big(\sum_{j\in\I} \x_{j,t}\Big) -u_\i(\xx_\i) .
\end{align}
composed of the \emph{network cost} and her utility, where $\xx_{-i} \eqDef (\xx_j)_{j\neq \i}$.
The instance $\G$  can be written as the tuple: \begin{equation}
\G=\left( \I, \T, \X , \cc, (u_i)_i\right) \ ,\end{equation}
 where $\X:=\X_1\times\dots \times \X_I$ and $\cc = (c_\t )_{\t\in\T}$.
\end{definition}
  In the remaining of this paper, the notation $\G$ will be used for an instance of an atomic game (\Cref{def:atomicGame}).

Owing to the network cost structure \eqref{eq:cost_player_def}, the aggregated load plays a central role. We denote it by $\xag_\t:= \sum_{\i\in\I} \x\iti$ on each arc $\t$,  and denote the associated feasibility set by:
\begin{equation} \label{eq:feas_set_agg_atomic}
\Sxag\eqDef \left\{ \xxag \in \rit^T \ :\ \exists \xx \in \X \text{ s.t. } \txt\sum_{\i\in\I} \xx_\i= \xxag  \right\}. 
\end{equation}

As seen in \eqref{eq:cost_player_def}, atomic splittable routing games are particular cases of aggregative games: each player's cost function depends on the actions of the others only through the aggregated profile $\xxag$.  

For technical simplification, we make the following assumptions:
\begin{assumption} \textbf{Convex costs} \label{assp_convex_costs}
Each cost function $(c_t)$ is differentiable, convex and  increasing. 
\end{assumption}
\begin{assumption} \textbf{Compact strategy sets} \label{assp_compactness}
For each $\i\in\I$, the set $\X_\i$ is assumed to be nonempty, convex and compact.
\end{assumption}

\begin{assumption}  \textbf{Concave utilities} \label{assp_concav_ut}
Each utility function $u_\i$ is differentiable and concave.
\end{assumption}

Note that under \Cref{assp_convex_costs,assp_concav_ut}, each function $f_\i$ is convex in $\xx_\i$.

An example that has drawn a particular attention is the class of atomic splittable routing games considered in \cite{orda1993competitive}. We add player-specific constraints on individual loads on each link, so that the model becomes the following.
\begin{example} \label{ex:splittable_feas_set}Each player $i$ has a weight $E_i$ to split over $\T$. In this case, $\X_\i$ is given as the simplex: \begin{equation*}
\X_\i=\{\ \xx_\i \in \rit^T_+ \ : \ \txt\sum_\t \x\iti= E_\i  \text{ and }  \underline{x}\iti \leq \x\iti \leq \overline{x}\iti\} \ .
\end{equation*}
 $E_i$ can be the mass of data to be sent over different canals, or an energy to be consumed over a set of time periods \cite{PaulinArxivAnalysisImpl17}. In the energy applications, more complex models include for instance ``ramping'' constraints $ \underline{r}\iti \leq\x_{\i,\t+1}-\x\iti \leq \overline{r}\iti $.
 \end{example}
 
\begin{example} \label{ex:utility}
An important example of utility function is the distance to a preferred profile $\yy_\i= (y\iti)_{\t \in \T}$, that is:
\begin{equation}
u_i(\xx_\i)= -\omega_i \norm{\xx_\i-\yy_\i}_2^2= -\omega_i\txt\sum_t\left( \x\iti-y\iti\right)^2 , 
\end{equation}
where $\omega_i>0$ is the value of player $\i$'s preference. 
Another type of utility function which has found many applications is :
\begin{equation} u_i(\xx_\i)= -\omega_i\log\left(1+\txt\sum_{t} \x\iti \right) \ , 
\end{equation}
which increases with the weight player $i$ can load on $\T$.

\end{example}
Below we recall the central notion of Nash Equilibrium in atomic non-cooperative games. 

\begin{definition}\textbf{{Nash Equilibrium (NE) }}

An $\NE$ of the atomic game $\G=\left( \I, \X , (f_i)_i\right)$ is a profile $\hxx \in \X$ such that for each player $\i\in \I$:
$$f_\i(\hxx_\i,\hxx_{-\i}) \leq f_\i(\xx_\i,\hxx_{-\i}), \ \forall \xx_i \in \X_i \ . $$
\end{definition}

\begin{proposition}\label{prop:GNNash} \textbf{Variational Formulation of an NE \\}
Under \Cref{assp_compactness,assp_convex_costs,assp_concav_ut}, $\hxx\in\X$ is an $\NE$ of $\G$ if and only if:
\begin{align} \label{cond:ind_opt_N}
\forall \xx \in\X, \forall\i \in \I, & \quad \big\langle \nabla_\i f_\i(\hxx_\i,\hxx_{-\i}), \xx_\i- \hxx_\i\big\rangle \geq 0 \ ,
 \end{align}
where $\nabla_\i f_\i(\hxx_\i,\hxx_{-\i})= \nabla f_\i(\cdot, \hxx_{-\i}) |_{\cdot = \hxx_{\i}}= \big(c_\t(\hxag_\t)+\hx\iti {c_\t}'(\hxag_\t)\big)_{t\in\T}-\nabla u_i(\hxx_i) $.
An equivalent condition is: 
 \begin{equation*}\label{cond:ind_opt_Nbis}
\forall \xx \in\X,\txt\sum_{\i\in \I}  \big\langle \nabla_\i f_\i(\hxx_\i,\hxx_{-\i}) , \xx_\i- \hxx_\i\big\rangle\geq 0 \ .
 \end{equation*}
\end{proposition}
\begin{proof}
Since $\xx_\i \mapsto f_\i(\xx_\i,\xx_{-\i})$ is convex, \eqref{cond:ind_opt_N} is the necessary and sufficient first order condition for $\hxx_\i$ to be a minimum of $f_\i(.,\hxx_{-\i})$.
\end{proof}
\vspace{-0.2cm}
 \Cref{def:atomicGame} defines a convex minimization game so that the existence of an NE is a corollary of Rosen's results \cite{rosen1965existence}:
\begin{theorem}[Cor. of Rosen, 1965]\textbf{Existence of an NE} \\
If $\G$ is an atomic routing congestion game (\Cref{def:atomicGame}) satisfying \Cref{assp_compactness,assp_convex_costs,assp_concav_ut}, then there exists an NE of $\G$.
\end{theorem}

Rosen \cite{rosen1965existence} gave a uniqueness theorem applying to any convex compact strategy sets, relying on a strong monotonicity condition of the operator $(\nabla_{\xx_\i} f_\i)_\i$. For  atomic splittable routing games \cite{orda1993competitive}, an NE is not unique in general \cite{bhaskar2009notunique}. To our knowledge, for atomic \emph{parallel} routing games (\Cref{def:atomicGame}) under \Cref{assp_compactness,assp_convex_costs,assp_concav_ut},  neither the uniqueness of NE nor a counter example of its uniqueness has been found. However, there are some particular cases where uniqueness has been shown, \textit{e.g.} \cite{PaulinArxivAnalysisImpl17} for the case of \Cref{ex:splittable_feas_set}. 

%
%
%

However, as we will see in the convergence theorems of \Cref{sec:approx_games},  uniqueness of NE is not  necessary to ensure the convergence of $\NE$ of a sequence of atomic unsplittable games, as any sequence of NE will converge to the unique Wardrop Equilibrium of the nonatomic game considered.


\subsection{Infinity of Players: the Nonatomic Framework}

If there is an infinity of players, the structure of the game changes: the action of a single player has a negligible impact on the aggregated load on each link. To measure the impact of infinitesimal players, we equip real coordinate spaces $\rit^k$ with the usual Lebesgue measure $\mu$.

The set of players is now represented by a continuum $\Theta=[0,1]$. Each player is of Lebesgue measure 0. 

\begin{definition} \textbf{Nonatomic Routing Game} \\ \label{def:nonatomicGame}
An instance $\Gna$ of a nonatomic routing game is defined by:\begin{itemize}[leftmargin=*,wide,labelindent=-1pt]
\item a continuum of players $\Theta=[0,1]$,
\item a finite set of arcs $\T=\{1,\dots,T\}$,
\item a point-to-set mapping of feasibility sets $\X_. : \Theta\rightrightarrows\rit_+^T$,
\item for each $\th\in\Theta$, a utility function $u_\th(.): \X_\th \rightarrow \rit$,
\item for each $\t\in\T$, a cost or latency function $c_t(.): \rit \rightarrow \rit$.
\end{itemize}

Each nonatomic player $\th$ chooses a profile $\xx_\th=(\x\thti)_{\t\in\T}$ in her feasible set $\X_\th$ and minimizes her cost function:
\begin{equation} \label{eq:user_cost_nonatomic}
\F_\th(\xx_\th,\xxag):=\sum_{\t\in\T}\x\thti c_t\Big(\xag_\t\Big) -u_\th(\xx_\th) ,
\end{equation}
where $\xag_\t\eqDef \int_{\Theta} \xx\thti \mathrm{d}\th$ denotes the aggregated load. The nonatomic instance $\Gna$  can be written as the tuple:
\begin{equation}
\Gna=\left( \Theta, \T, (\X_\th)_{\th\in\Theta} , \cc, (u_\th)_{\th\in\Theta}\right) \ .
\end{equation}
\end{definition}

For the nonatomic case, we need assumptions stronger than \Cref{assp_compactness,assp_concav_ut} for the mappings $\X_.$ and $u.$, given below:

\begin{assumption} \label{ass_X_nonat}\textbf{Nonatomic strategy sets} There exists $\diamX>0$ such that, for any $\th \in \Theta$, $\X_\th$ is convex, compact and $\X_\th \subset \mathcal{B}_0(M)$, where $\mathcal{B}_0(M)$ is the ball of radius $M$ centered at the origin. Moreover, the mapping $\th \mapsto \ \X_\th$  has a measurable graph $\Gamma_\X \eqDef\{(\th,\xx) : \th \in \Theta, \xx \in \X_\th  \}\subset \rit^{T+1}$.
\end{assumption}

\begin{assumption} \textbf{Nonatomic utilities} \label{ass_ut_nonat}  There exists $\Bcut>0$ s.t. for each $\th$, $u_\th$ is  differentiable, concave and $\norm{\nabla u_\th}_\infty < \Bcut$. The function $\Gamma_\X \ni (\th,\xx_\th) \mapsto u_\th(\xx_\th)$ is measurable.
\end{assumption}

%

\Cref{def:nonatomicGame} and \Cref{ass_X_nonat,ass_ut_nonat} give a very general framework. In many models of nonatomic games that have been considered, players are considered homogeneous  or with a finite number of classes \cite[Chapter 18]{nisan2007algorithmic}. Here, players can be heterogeneous through $\X_\th$ and $u_\th$. Games with heterogeneous players can find many applications, an example being the nonatomic equivalent of \Cref{ex:splittable_feas_set}: 
\begin{example} Let $\th \mapsto E_\th$ be a density function which designates the total demand $E_\th$ for each player $\th\in\Theta$. Consider the \emph{nonatomic splittable routing game} with feasibility sets $$\X_\th\eqDef \{\xx_\th \in\rit_+^T : \ \sum_t \x_{\th,t}= E_\th \} . $$ 
As in  \Cref{ex:splittable_feas_set}, one can consider some upper bound $\overline{\x}_{\th,t}$ and lower bound $\underline{\x}_{\th,t}$ for each $\th \in \Theta$ and each $\t\in \T$, and add the bounding constraints $\forall \t \in \T, \underline{\x}_{\th,\t}  \leq \x_{\th,\t} \leq \overline{\x}_{\th,t}$ in the definition of $\X_\th$.
\end{example}
Heterogeneity  of utility functions can also appear in many practical cases: if we consider the case of \emph{preferred profiles} given in \Cref{ex:utility}, members of a population can attribute different values to their cost and their preferences.

Since each player is infinitesimal, her action has a negligible impact on the other players' costs. Wardrop \cite{wardrop1952some} extended the notion of equilibrium to the nonatomic case.

\begin{definition} \textbf{Wardrop Equilibrium (WE) } \\
$\xx^*\in (\X_\th)_{\th}$ is a Wardrop equilibrium of the game $\Gna$ if it is a \textbf{measurable} function from $\th$ to $\X$  and for almost all $\th \in \Theta$,
\begin{equation*}
 \F_\th (\xx_\th^*, \xxag^*) \leq \F_\th (\xx_\th, \xxag^*), \ \forall \xx_\th \in \X_\th  \ ,
\end{equation*}
where $\xxag^*=\int_{\th\in \Theta}\xx^*_\th \mathrm{d}\th \in \rit^T$.
\end{definition}

\begin{proposition} \textbf{Variational formulation of a WE}\\\label{prop:cond_nash_wardrop}
Under \Cref{assp_convex_costs,ass_X_nonat,ass_ut_nonat}, $\xx^* \in \X$ is a WE of $\Gna$ iff for almost all $\th \in \Theta$:
\begin{equation}\label{cond:ind_opt}
\langle c(\xxag^*) - \nabla u_\th(\xx_\th^*), \xx_\th - \xx_\th^* \rangle \geq 0, \quad \forall \xx_\th \in \X_\th \ .
 \end{equation}
\end{proposition}
\begin{proof}\
Given $\xxag^*$, \eqref{cond:ind_opt} is the necessary and sufficient first order condition for $\xx^*_\th$ to be a minimum point of the convex function $\F_\th (., \xxag^*)$.
\end{proof}

According to \eqref{cond:ind_opt}, the monotonicity of $c$ is sufficient to have the VI characterization of the equilibrium in the nonatomic case, as opposed to the atomic case in \eqref{cond:ind_opt_N} where monotonicity \emph{and} convexity of $c$ are needed.
\begin{theorem}[Cor. of Rath, 1992 \cite{rath1992direct}] \textbf{Existence of a WE} \\
If $\Gna$ is a nonatomic routing congestion game (\Cref{def:nonatomicGame}) satisfying \Cref{assp_convex_costs,ass_X_nonat,ass_ut_nonat}, then $\Gna$ admits a WE.
%
%
%
\end{theorem}
\begin{proof} 
The conditions required in \cite{rath1992direct} are satisfied. Note that we only need $(c_t)_t$ and $(u_\th)_{\th\in\Theta}$ to be continuous functions.\end{proof}
The variational formulation of a WE given in \Cref{prop:cond_nash_wardrop} can be written in the closed form:
\begin{theorem}\label{thm:agg_wardrop}
Under \Cref{assp_convex_costs,ass_X_nonat,ass_ut_nonat}, $\xx^* \in \X$ is a WE of $\Gna$ iff:
\begin{equation}\label{cond:agg_eq}
\int_{\th\in \Theta} \langle \cc(\xxag^*) - \nabla u_\th(\xx_\th^*), \xx_\th - \xx_\th^* \rangle \mathrm{d}\th \geq 0,\quad \forall \xx\in \X \ .
 \end{equation} 
\end{theorem}
\begin{proof} This follows from \Cref{prop:cond_nash_wardrop}. 
\ifproofs
If $\xx^*\in \X$ is a Wardrop equilibrium so that \eqref{cond:ind_opt} holds for almost all $\th\in \Theta$, then \eqref{cond:agg_eq} follows straightforwardly.

Conversely, suppose that $\xx^*\in \X$ satisfies condition \eqref{cond:agg_eq} but is not a WE of $G$. Then there must be a subset $S$ of $\Theta$ with strictly positive measure such that for each $\th \in S$, \eqref{cond:ind_opt} does not hold: for each $\th \in S$, there exists $\yy_\th \in \X_\th$ such that 
\begin{equation*}
\langle \cc(\xxag^*) - \nabla u_\th(\xx_\th^*), \yy_\th - \xx_\th^* \rangle < 0
 \end{equation*}
 For each $\th\in \Theta\setminus S$, let $\yy_\th \eqDef \xx^*_\th$. Then $\yy=(\yy_\th)_{\th\in \Theta}\in \X$, and
\begin{align*}
& \int_{\th\in \Theta} \langle \cc(\xxag^*) - \nabla u_\th(\xx_\th^*), \yy_\th - \xx_\th^* \rangle \text{d}\th\\
&= \int_{\th\in S} \langle \cc(\xxag^*) - \nabla u_\th(\xx_\th^*), \yy_\th - \xx_\th^* \rangle \text{d}\th <0  
\end{align*}
contradicting \eqref{cond:agg_eq}.
\fi
\end{proof}

\begin{corollary}\label{cor:agg_wardrop_wou}
In the case where $u_\th\equiv 0$ for all $\th\in \Theta$,  under \Cref{assp_convex_costs,ass_X_nonat}, $\xx^*\in \X$ is a WE of $\Gna$ iff:
\begin{equation}\label{cond:agg_eq_wou}
 \langle \cc(\xxag^*) , \xxag - \xxag^* \rangle  \geq 0, \quad \forall \xxag\in \tilde{\X}\ .
 \end{equation} 
\end{corollary}

From the characterization of the WE in \Cref{thm:agg_wardrop} and \Cref{cor:agg_wardrop_wou}, we derive \Cref{thm:unique_wardrop,th:unique_agg_WE} that state simple  conditions ensuring the uniqueness of WE in $\Gna$.

\begin{theorem}\label{thm:unique_wardrop}
Under \Cref{assp_convex_costs,ass_X_nonat,ass_ut_nonat}, if $u_\th$ is strictly concave for each $\th \in \Theta$, then $G$ admits a unique WE.
\end{theorem}
\ifproofs
\begin{proof}

Suppose that $\xx\in \X$ and $\yy\in \X$ are both WE of the game. Let $X=\int_{\th\in \Theta}\xx_\th \text{d}\th$ and $Y=\int_{\th\in \Theta}\yy_\th \text{d}\th$. Then, according to Theorem \ref{thm:agg_wardrop}, 
\begin{align}
& \int_{\th\in \Theta} \langle \cc(\xxag) - \nabla u_\th(\xx_\th), \yy_\th - \xx_\th \rangle \text{d}\th \geq 0 \label{eq:1}\\
& \int_{\th\in \Theta} \langle \cc(\yyag) - \nabla u_\th(\yy_\th), \xx_\th - \yy_\th \rangle \text{d}\th \geq 0 \label{eq:2}
\end{align}
By adding \eqref{eq:1} and \eqref{eq:2}, one has
\begin{align*}
  \int_{\th\in \Theta} & \langle \cc(\xxag)-\cc(\yyag) - \nabla u_\th(\xx_\th)+\nabla u_\th(\yy_\th), \yy_\th - \xx_\th \rangle \text{d}\th \geq 0\\
 \Rightarrow\; & \langle \cc(\xxag)-\cc(\yyag), \int_{\th\in \Theta}(\yy_\th - \xx_\th) \text{d}\th \rangle  + \int_{\th\in \Theta} \langle  - \nabla u_\th(\xx_\th)+\nabla u_\th(\yy_\th), \yy_\th - \xx_\th \rangle \text{d}\th \geq 0\\
 \Rightarrow\; & \langle \cc(\xxag)-\cc(\yyag), X-Y\rangle  + \int_{\th\in \Theta}\langle  - \nabla u_\th(\xx_\th)+\nabla u_\th(\yy_\th), \xx_\th - \yy_\th \rangle \text{d}\th \leq 0
\end{align*}
Since for each $\th$, $u_\th$ is strictly concave,  $\nabla u_\th$ is thus strictly monotone. Therefore, for each $\th\in \Theta$, $ \langle - \nabla u_\th(\xx_\th)+\nabla u_\th(\yy_\th), \xx_\th - \yy_\th \rangle \geq 0$ and equality holds if and only if $\xx_\th=\yy_\th$. Besides, $c$ is monotone, hence $ \langle \cc(\xxag)-\cc(\yyag), X-Y\rangle\geq 0$. Consequently, $\langle \cc(\xxag)-\cc(\yyag), X-Y\rangle + \int_{\th\in \Theta}\langle  - \nabla u_\th(\xx_\th)+\nabla u_\th(\yy_\th), \xx_\th - \yy_\th \rangle \text{d}\th \geq 0$, and equality holds if and only if for almost all $\th\in \Theta$, $\xx_\th=\yy_\th$. (In this case, $X=Y$.)
\end{proof}
\else
\fi

\begin{theorem} \label{th:unique_agg_WE}
In the case where $u_\th\equiv 0$ for all $\th\in \Theta$, under \Cref{assp_convex_costs,ass_X_nonat},  if $\cc=(c_\t)_{t=1}^T: [0,M]^T \rightarrow \rit^T$ is a strictly monotone operator, then all the WE of $\Gna$ have the same aggregate profile $\xxag^* \in \Sxag$.
\end{theorem}
\ifproofs
\begin{proof}
Suppose that $\xx\in \X$ and $\yy\in \X$ are both WE of the game. Let $\xxag=\int_{\th\in \Theta}\xx_\th \text{d}\th$ and $\yyag=\int_{\th\in \Theta} \yy_\th \text{d}\th$. Then, according to Corollary \ref{cor:agg_wardrop_wou}, 
\begin{align}
&  \langle \cc(\xxag) , \yyag - \xxag \rangle  \geq 0 \label{eq:1b}\\
&  \langle \cc(\yyag) , \xxag - \yyag \rangle  \geq 0 \label{eq:2b}
\end{align}
By adding \eqref{eq:1b} and \eqref{eq:2b}, one has
\begin{equation*}
  \langle \cc(\xxag)-\cc(\yyag) , \yyag -\xxag \rangle  \geq 0\\
\end{equation*}
Since $c$ is strictly monotone, $ \langle \cc(\xxag)-\cc(\yyag), \xxag-\yyag\rangle\geq 0$ and equality holds if and only $\xxag=\yyag$. Consequently, $\xxag=\yyag$.
\end{proof}
\else
\fi
\begin{remark}
If for each $t \in \T$, $c_\t(.)$ is (strictly) increasing, then $\cc$ is a (strictly) monotone operator from $[0,M]^T \rightarrow \rit^T$.
\end{remark}

One expects that, when the number of players grows very large in an atomic splittable game, the game gets close to a nonatomic game in some sense. We confirm this intuition by showing that, considering a sequence of equilibria of approximating atomic games of a nonatomic instance, the sequence will converge to an equilibrium of the nonatomic instance.
 
\section{Approximating Nonatomic Games}
\label{sec:approx_games}
\newcommand{\nn}{\bm{n}}
\newcommand{\snu}{\nu} 
\newcommand{\esnu}{^{(\snu)}}
\newcommand{\dset}{\delta} 
\newcommand{\mdset}{\overline{\delta}} 

\newcommand{\duti}{d} 
\newcommand{\mduti}{\overline{d}} 
\newcommand{\mmu}{\overline{\mu}} 

To approximate the nonatomic game $\Gna$, the idea consists in finding a sequence of atomic games $(\G\esnu)$ with an increasing number of players, each player representing a ``class'' of nonatomic players, similar in their parameters. 

As the players $\th\in\Theta$ are differentiated through  $\X_\th$ and  $u_\th$, we need to formulate the convergence of feasibility sets and utilities of atomic instances to the nonatomic parameters.

\subsection{Approximating the nonatomic instance}

\begin{definition} \label{def:approx_seq} \textbf{Atomic Approximating Sequence (AAS)} \\
A sequence of atomic games $\G\esnu=\big(\I\esnu,\T,\X\esnu ,\cc,(u_i\esnu)_i\big)$ is an approximating sequence (AAS) for the nonatomic instance $\Gna=\big(\Theta,\T,(\X_\th)_\th,\cc,(u_\th)_\th\big)$ if for each $\snu \in \nit$, there exists a partition  of cardinal $I\esnu$ of  set $\Theta$, denoted by $(\Theta_\i\esnu)_{i\in \I\esnu}$, such that:
\begin{itemize}[leftmargin=*,wide,labelindent=-1pt]
\item $I\esnu \longrightarrow +\infty$,
\item $\mmu\esnu \eqDef\max_{\i \in\I\esnu}\mu_i\esnu \longrightarrow 0$  where $\mu_\i\esnu\eqDef\mu(\Theta_i\esnu)$ is the Lebesgue measure of subset $\Theta_i\esnu$,
\item $\mdset\esnu\eqDef\max_{\i \in\I\esnu} \dset_\i\esnu \longrightarrow 0$ where $\dset_\i$ is the Hausdorff distance (denoted by $d_H$) between nonatomic feasibility sets and the scaled atomic feasibility set:
\begin{equation} \label{eq:def_dset}
\dset_\i\esnu \eqDef \max_{\th \in \Theta_\i} d_H\left( \X_\th, \txt\frac{1}{\mu_i\esnu}\X_{\i}\esnu \right),
\end{equation}
\item $\mduti\esnu\eqDef\max_{\i\in\I\esnu}\duti_\i\esnu \longrightarrow 0  $ where $\duti_\i$ is the $L_\infty$-distance (in $\mathcal{B}_0(\diamX) \rightarrow \rit$) between the gradient of nonatomic utility functions and the scaled atomic utility functions:
\begin{equation}\label{eq:def_dut}
\duti\esnu_\i= \max_{\th \in \Theta_\i} \max_{\xx\in\B_0(M)} \norm{ \nabla u\esnu_\i\left(\txt {\mu_i\esnu } \xx \right) - \nabla u_\th(\xx)}_2 \ .
\end{equation} 

\end{itemize} 
\end{definition}

From \Cref{def:approx_seq} it is not trivial to build an AAS of a given nonatomic game $G$, one can even  be unsure that such a sequence exists. However, we will give practical examples in \Cref{subsec:approx_mesh,subsec:approx_uniform}.

A direct result from the assumptions in \Cref{def:approx_seq} is that the players become infinitesimal, as stated in \Cref{lem:norm_xat_approx}.
\begin{lemma}\label{lem:norm_xat_approx}
If $(\G\esnu)_\nu$ is an AAS of a nonatomic instance $\Gna$, then considering the maximal diameter $\diamX$ of $\X_\th$, we have:
\begin{equation}
\forall \i \in \I\esnu, \forall \xx_i\in\X_\i\esnu , \ \ \norm{\xx_i}_2 \leq \mu_i\esnu (\diamX+\dset_i\esnu) \ . 
\end{equation}
\end{lemma}
\begin{proof}
\ifproofs
Let $\xx_\i \in \X_\i\esnu$. Let $\th \in \Theta_\i\esnu$ and denote by $P_{\X_\th}$  the projection on $\X_\th$.
By definition of $\dset_\i\esnu$, we get:
\begin{align}
& \norm{\frac{\xx_i}{\mu_\i\esnu} -P_{\X_\th}\Big(\frac{\xx_i}{\mu_\i\esnu}\Big)}_2 \leq  \dset_\i\esnu \\
\Longleftrightarrow \ & \norm{\xx_i}_2 \leq \mu_i\esnu \left( \dset_\i\esnu  + \norm{P_{\X_\th}\Big(\frac{\xx_i}{\mu_\i\esnu}\Big)}_2 \right) \leq \mu_i\esnu (\dset_\i\esnu  + M) \ . 
\end{align}
\fi
\end{proof}
\begin{lemma}\label{lem:hausdorff_agg_sets}
If $(\G\esnu)_\nu$ is an AAS of a nonatomic instance $\Gna$, then the Hausdorff distance between the aggregated sets $\Sxag=\int_\Theta \X_.$ and $\Sxag\esnu=\sum_{\i\in\I\esnu} \X_\i\esnu$ is bounded by:
\begin{equation}
d_H\left( \Sxag\esnu,\Sxag \right) \leq \mdset\esnu \ .
\end{equation}
\end{lemma}
\begin{proof}
\ifproofs
Let $(\xx_\th)_\th \in \X$ be a nonatomic profile. Let $P_i$ denote the Euclidean projection on $\X_\i\esnu$ for $\i\in\I\esnu$ and consider $\yy_i\eqDef P_i\left(\int_{\Theta_i\esnu}  \xx_\th \dth \right)  \in \X_\i\esnu$. From \eqref{eq:def_dset} we have:
\begin{align}
\norm{ \int_{\Theta} \xx_\th \dth - \sum_{i\in\I\esnu}\yy_i}_2 & =\norm{ \sum_{i\in\I\esnu} \left( \int_{\Theta_\i\esnu} \xx_\th \dth - \yy_i \right) }_2  \\
& =\norm{ \sum_{i\in\I\esnu}  \int_{\Theta_\i\esnu} \left( \xx_\th - \frac{1}{\mu_\i\esnu}\yy_i \right) \dth  }_2  \\
& \leq \sum_{i\in\I\esnu} \int_{\Theta_\i\esnu} \norm{ \xx_\th - \frac{1}{\mu_\i\esnu}\yy_i }_2 \dth \\
& \leq \sum_{i\in\I\esnu} \int_{\Theta_\i\esnu}  \dset_\i\esnu \dth = \sum_{i\in\I\esnu} \mu_\i\esnu \dset_\i\esnu \leq \mdset\esnu \ ,
\end{align}
which shows that $d \left( \xxag, \Sxag\esnu \right) \leq \mdset\esnu$ for all $\xxag\in\ \Sxag$. On the other hand, if $\sum_{\i\in\I\esnu} \xx_i \in \Sxag\esnu$, then let us denote by $\Pi_\th$ the  Euclidean projection on $\X_\th$ for $\th \in \Theta$, and $\yy_\th= \Pi_\th\left(\frac{1}{\mu_\i\esnu} \xx_i \right) \in \X_\th$ for $\th \in \Theta_i\esnu$. Then we have for all $\th\in\Theta_\i\esnu$, $\norm{\frac{1}{\mu_\i\esnu}\xx_i - \yy_\th }_2 \leq \dset_\i\esnu $ and we get:
\begin{align}
\norm{\sum_{\i\in\I\esnu} \xx_i  - \int_{\Theta } \yy_\th \dth }_2 & \leq \sum_{\i\in\I\esnu}  \norm{\int_{\Theta_\i\esnu}\frac{1}{\mu_\i\esnu}  \xx_i  - \yy_\th \dth   }_2 \\ 
& \leq  \sum_{\i\in\I\esnu}  \int_{\Theta_\i\esnu} \norm{ \frac{1}{\mu_\i\esnu}  \xx_i  - \yy_\th }_2 \dth  \\
& \leq  \sum_{i\in\I\esnu} \mu_\i\esnu \dset_\i\esnu \leq \mdset\esnu \ ,
\end{align}
which shows that $d \left( \xxag, \Sxag\right) \leq \mdset\esnu$ for all $\xxag\in\ \Sxag\esnu$ and concludes the proof.
\fi
\end{proof}
To ensure the convergence of an AAS, we make the following additional assumptions on costs functions $(c_\t)_t$:
\begin{assumption} \textbf{Lipschitz continuous costs}\label{assp_lip}
For each $t\in \T$, $c_\t$ is a Lipschitz continuous function on $[0,M]$. There exists $C>0$ such that for each $t\in\T$, $|{c_\t}'(\cdot)|\leq C$.
\end{assumption}


\begin{assumption} \textbf{Strong monotonicity}\label{assp_strongmon}
There exists $\stmon>0$ such that, for each $t\in\{1,\ldots,T\}$, ${c_\t}'(\cdot)\geq \stmon$ on $[0,M]$
\end{assumption}

In the following sections, we differentiate the cases with and without utilities, because we found different convergence results in the two cases.

\subsection{Players without Utility Functions: Convergence of the Aggregated Equilibrium Profiles}\label{sec:atomic_without_u}
In this section, we assume that $u_\th\equiv 0$ for each $\th \in \Theta$.

We give a first result on the approximation of WE by a sequence of NE in \Cref{thm:converge_without_u}.

\begin{theorem}\label{thm:converge_without_u} Let $(\G\esnu)_\snu$ be an AAS of a nonatomic instance $\Gna$, satisfying \Cref{assp_compactness,assp_convex_costs,ass_X_nonat,assp_lip,assp_strongmon}. Let $(\hxx\esnu)$ a sequence of $\NE$ associated to $(\G\esnu)$, and $(\xx^*_\th)_\th$ a WE of $\Gna$. Then:
\begin{equation*} 
\norm{ \hxxag\esnu- \xxag^* }_2^2 \leq \frac{2}{\stmon} \times \left(\Bc \times \mdset\esnu + C(M+1)^2\times \mmu\esnu \right)\ , 
\end{equation*}
where $\Bc\eqDef\max_{\xx\in\B_0(M)} \norm{\cc(\xx)}_2$.
\end{theorem}
\ifproofs
\begin{proof}
Let $P_i$ denote the Euclidean projection onto $\X_\i\esnu$ and $\Pi$ the projection onto $\Sxag$. We omit the index $\nu$ for simplicity. From \eqref{cond:agg_eq_wou}, we get:
\begin{equation} \label{eq:VI_WEwithNE}
\left\langle \cc(\xxag^*), \Pi(\hxxag)-\xxag^* \right\rangle \geq 0 \ .
\end{equation}
On the other hand, with $\xx_\i^*\eqDef\int_{\Theta_i} \xx_\th \mathrm{d}\th$, we get from \eqref{cond:ind_opt_Nbis}:
\begin{align} \label{eq:VI_NEwithWE}
0 \leq &\sum_{\i\in \I}  \big\langle \big(c_\t(\hxag_\t)+\hx\iti {c_\t}'(\hxag_\t)\big)_{t\in\T} , P_i(\xx^*_\i)- \hxx_\i\big\rangle \\
=&\left\langle \cc(\hxxag),  \txt\sum_i P_i(\xx^*_\i)- \hxxag \right\rangle + R(\hxx,\xx^*) 
\end{align}
with  $R(\hxx,\xx^*) =\sum_i  \big\langle \big(\hx\iti {c_\t}'(\hxag_\t)\big)_{t} , P_i(\xx^*_\i)- \hxx_\i\big\rangle $.
From the Cauchy-Schwartz inequality and  \Cref{lem:norm_xat_approx}, we get:
\begin{align}
|R(\hxx,\xx^*) |& \leq \sum_{\i\in\I\esnu} \norm{ \big(\hx\iti {c_\t}'(\hxag_\t)\big)_{t} }_2 \times  \norm{P_i(\xx^*_\i)- \hxx_\i  }_2 \\
& \leq \sum_{\i\in\I\esnu} ( \mu_i\esnu (\diamX+\dset_i\esnu) C\times 2 ( \mu_i\esnu (\diamX+\dset_i\esnu) ) \\  
&\leq 2C(M+1)^2 \max_i \mu_\i\esnu \ .
\end{align}
Besides, with the strong monotonicity of $\cc$ and from \eqref{eq:VI_WEwithNE} and \eqref{eq:VI_NEwithWE}:
\begin{align*}
&\stmon\norm{\hxxag-\xxag^* }^2 \leq \left\langle \cc(\hxxag)-\cc(\xxag^*), \hxxag-\xxag^* \right\rangle \\
&=\big\langle \cc(\hxxag), \hxxag-\xxag^* \big\rangle  +\big\langle \cc(\xxag^*), \xxag^*-\hxxag \big\rangle \\
\nonumber&\leq \big\langle \cc(\hxxag), \hxxag- \sum_i P_i(\xx^*_\i)\big\rangle  +\big\langle \cc(\xxag^*), \xxag^*-\Pi(\hxxag) \big\rangle   \\ 
& +  \big\langle \cc(\hxxag), \sum_i P_i(\xx^*_\i)-\xxag^*\big\rangle  +\big\langle \cc(\xxag^*), \Pi(\hxxag)-\hxxag  \big\rangle \\
&\leq |R(\hxx,\xx^*) | + 0 + 2\Bc \times \max_i \dset_i \ ,
\end{align*}
which concludes the proof.
\end{proof}
\fi

\subsection{Players with Utility Functions: Convergence of the Individual Equilibrium Profiles}\label{sec:atomic_with_u}

%
%
\newcommand{\liput}{L}
In order to establish a convergence theorem in the presence of utility functions, we make an additional assumption of strong monotonicity on the utility functions stated in \Cref{assp7_stgconcav_util}.
Note that this assumption 
 holds for the  utility functions given in \Cref{ex:utility}.

\newcommand{\stgccvut}{{\alpha}} 

\begin{assumption} \textbf{Strongly concave utilities}\label{assp7_stgconcav_util}
For all $\th\in \Theta$, $u_\th$ is strongly concave on $\B_0(\diamX)$, uniformly in $\th$: there exists $\stgccvut>0$ such that for all  $\xx,\yy\in \B_0(\diamX)^2$ and any $ \tau \in ]0,1[$ :
\vspace{-0.2cm}
\begin{equation*}
\hspace{-4pt} u_\th((1-\tau)\xx+\tau \yy)\geq(1-\tau) u_\th(\xx) + \tau u(\yy) + \textstyle\frac{\stgccvut}{2} \tau (1-\tau) \|\xx-\yy\|^2.
\end{equation*}
\end{assumption}


\begin{remark}
If $u_\th( \xx_\th)$ is $\alpha_\th$-strongly concave, then the negative of its gradient is a strongly monotone operator: 
\begin{equation}\label{eq:minusVmono}
 - \langle \nabla u_\th( \xx_\th) - \nabla u_\th ( \yy_\th), \xx_\th -\yy_\th \rangle \geq \alpha_\th\|\xx_\th - \yy_\th \|^2  \ .
\end{equation}  
\end{remark}

We start by showing that, under the additional \Cref{assp7_stgconcav_util} on the utility functions, the WE profiles of two nonatomic users within the same subset $\Theta_i\esnu$ are roughly the same.

\newcommand{\distinx}{r}

\begin{proposition} \label{lem:dist_profile_inside_part}
Let $(\G\esnu)_\snu$ be an AAS of a nonatomic instance $\Gna$ and  $(\xx^*_\th)_\th$ the WE of $\Gna$ satisfying \Cref{ass_ut_nonat,assp_convex_costs,ass_X_nonat,assp7_stgconcav_util}. Then, if $\th, \xi \in \Theta_\i\esnu$, we have:
\begin{equation*} 
\hspace{-0.2cm}\norm{\xx^*_\th-\xx^*_\xi}_2^2 \leq \frac{2}{\stgccvut} \left(\diamX\duti_\i\esnu  + (\Bc+\Bcut) \dset_i\esnu \right) .
\end{equation*}

\end{proposition}
\ifproofs

\begin{proof}
Let $\yy_\th=P_{\X_\xi}(\xx^*_\th)$ and conversely $\yy_\xi=P_{\X_\th}(\xx^*_\xi)$.
Then from \Cref{prop:cond_nash_wardrop} we get:
 \begin{align}
 &\langle \cc(\xxag^*) - \nabla u_\th( \xx^*_\th), \yy_\xi +(\xx^*_\xi-\xx^*_\xi)- \xx^*_\th \rangle \geq 0 \label{eq:3B} \\
 & \langle \cc(\xxag^*) - \nabla u_\xi(\xx^*_\xi), \yy_\th +(\xx^*_\th-\xx^*_\th)- \xx^*_\xi \rangle \geq 0 \label{eq:4B} .  \end{align}
Denote by $v_\th=\nabla u_\th$ and $v_\xi=\nabla u_\xi$. Then, using the strong concavity \Cref{assp7_stgconcav_util} of $u_\th$, we get:
\begin{align}
\stgccvut &\norm{\xx^*_\th      -\xx^*_\xi}^2 \leq \langle v_\th(\xx^*_\th)-v_\th(\xx^*_\xi), \xx^*_\xi-\xx^*_\th \rangle \\
  & \leq \langle v_\th(\xx^*_\th)+v_\xi(\xx^*_\xi)-v_\xi(\xx^*_\xi)-v_\th(\xx^*_\xi), \xx^*_\xi-\xx^*_\th \rangle 
 \\ \label{eq:dist_tutil_boundproof}
& \leq \duti\esnu_i\norm{\xx^*_\xi-\xx^*_\th} +\langle v_\th(\xx^*_\th) - v_\xi(\xx^*_\xi),\xx^*_\xi-\xx^*_\th  \rangle \ .
\end{align}
Then,  \eqref{eq:3B} and \eqref{eq:4B} yield:
\begin{align*}
&\langle v_\th(\xx^*_\th) - v_\xi(\xx^*_\xi),\xx^*_\xi-\xx^*_\th  \rangle \\
& = \langle v_\th(\xx^*_\th) -\cc(\xxag),\xx^*_\xi-\xx^*_\th  \rangle +\langle v_\xi(\xx^*_\xi) -\cc(\xxag) ,\xx^*_\th  - \xx^*_\xi \rangle \\
& = \langle v_\th(\xx^*_\th) -\cc(\xxag),\xx^*_\xi-\yy_\xi  \rangle +\langle v_\xi(\xx^*_\xi) -\cc(\xxag) ,\xx^*_\th  - \yy_\th \rangle \\
& \leq \dset_\i\esnu \norm{v_\th(\xx^*_\th) -\cc(\xxag)} + \dset_\i\esnu \norm{v_\xi(\xx^*_\xi) -\cc(\xxag)} \ ,
\end{align*}
which gives the desired result when combined with \eqref{eq:dist_tutil_boundproof}.
\end{proof}
\fi

This result reveals the role of the strong concavity of utility functions: when $\stgccvut$ goes to $0$, the right hand side of the inequality diverges. This is coherent with the fact that, without utilities, only the aggregated profile matters, so that we cannot have a result such as \Cref{lem:dist_profile_inside_part}.

According to \Cref{lem:dist_profile_inside_part}, we can obtain a continuity property of the Wardrop equilibrium if we introduce the notion of continuity for the nonatomic game $\Gna$, relatively to its parameters:
\begin{definition} \textbf{Continuity of a nonatomic game } \\ \label{def:continuity_nonatomic_game}
The nonatomic instance $\Gna=\big(\Theta,\T,(\X_\th)_\th,\cc,(u_\th)_\th\big)$ is said to be continuous at $\th \in \Theta$ if, for all $\varepsilon>0$, there exists $\eta>0$ such that:
\begin{equation}
\forall  \th' \in\Theta, \  \|\th-\th'\| \leq \eta \ \Rightarrow \left\{ \begin{array}{l} d_H(\X_\th, \X_{\th'}) \leq \varepsilon  \\ \max_{\xx\in \X_\th \cup \X_{\th'}} \norm{\nabla u_\th(\xx) - \nabla u_{\th'}(\xx) }_2 \leq \varepsilon \end{array} \right. \ .
\end{equation}
\end{definition}

Then the proof of \Cref{lem:dist_profile_inside_part} shows the following intuitive property:
\begin{proposition}\label{prop:continuity_from_parameters}
Let $\Gna=\big(\Theta,\T,(\X_\th)_\th,\cc,(u_\th)_\th\big)$ be a nonatomic instance. If $\Gna$ is continuous at $\th_0\in \Theta$ and $(\xx^*_\th)_\th$ is a WE of $\Gna$, then $\th \mapsto \xx^*_\th$ is continuous at $\th_0$.
\end{proposition}

The next theorem is one of the main results of this paper. It shows that a WE can be approximated by the NE of an atomic approximating sequence.

\begin{theorem}\label{thm:converge_with_u}
 Let $(\G\esnu)_\snu$ be an AAS of a nonatomic instance $\Gna$. Let $(\hxx\esnu)$ a sequence of $\NE$ associated to $(\G\esnu)$, and $(\xx^*_\th)_\th$ the WE of $\Gna$. 
Under \Cref{assp_lip,assp_compactness,assp_concav_ut,assp7_stgconcav_util,ass_ut_nonat,assp_convex_costs,ass_X_nonat},  
  the approximating solution defined by $\hxx_\th\esnu:= \frac{1}{\mu_\i\esnu}\hxx_\i\esnu$ for $\th\in\Theta_\i\esnu$ satisfies:
  \vspace{-0.2cm}
 \begin{align*} 
&\int_{\th\in \Theta} \norm{ \hat{\xx}\esnu_\th - \xx^*_\th }^2_2 \mathrm{d}\th   \leq \frac{2}{ \stgccvut} \left(\Bc +\Bcut)\mdset\esnu + C (\diamX+ 1)^2  \mmu\esnu +\diamX\mduti\esnu \right)   .
\end{align*}
\end{theorem}

\ifproofs
\begin{proof}
\newcommand{\hyy}{\hat{\yy}}
Let $(\hxx_i)_i$ be an NE of $\G\esnu$, and $\xx^* \in \X $ the WE of $\Gna$. For the remaining of the proof we ommit the index $(\snu)$ for simplicity. 

Let us consider  the nonatomic profile defined by $\hxx_\th \eqDef \frac{1}{\mu_\i} \hxx_i$ for $\th \in \Theta_\i$, and its projection on the feasibility set $\hyy_\th\eqDef P_{\X_\th }(\hxx_\th) $. Similarly, let us consider the atomic profile given by $\xx^*_i\eqDef \int_{\Theta_i} \xx^*_\th \dth $ for $\i\in\I\esnu$, and its projection $\yy^*_i\eqDef P_{\X_i}(\xx_\i^*) $.

For notation simplicity, we denote $\nabla u_\th $ by $v_\th$.  From the strong concavity of $v_\th$ and the strong monotonicity of $\cc$, we have:
\begin{align}
& \stgccvut  \int_{\th\in \Theta} \norm{ \hat{\xx}\esnu_\th - \xx^*_\th }^2_2 + \stmon \norm{ \hxxag\esnu- \xxag^* }_2^2 \\
& \leq \int_{\Theta} \left\langle \cc(\hxxag)-v_\th(\hxx_\th) -\left(\cc(\xxag^*)- v_\th( \xx^*_\th) \right), \ \hxx_\th - \xx^*_\th  \right \rangle \dth \\
&=\int_{\Theta} \left\langle \cc(\hxxag)-v_\th(\hxx_\th)       , \ \hxx_\th - \xx^*_\th    \right \rangle \dth  + \int_{\Theta} \left\langle  \cc(\xxag^*)- v_\th( \xx^*_\th) , \ \xx^*_\th -\hxx_\th  \right \rangle \dth \ . \label{eq:proof_profile_ut_two_ints_tobound}
\end{align}
To bound the second term, we use the characterization of a WE  given in \Cref{prop:cond_nash_wardrop}, with $\hyy_\th \in\X_\th$ :
\begin{align}
& \int_{\Theta} \left\langle \cc(\xxag^*)- v_\th( \xx^*_\th) , \ \xx^*_\th  - \hxx_\th  \right \rangle \dth \\
& =\int_{\Theta} \left\langle \cc(\xxag^*)- v_\th( \xx^*_\th) , \ \xx^*_\th  - \hyy_\th  \right \rangle \dth + \int_{\Theta} \left\langle \cc(\xxag^*)- v_\th( \xx^*_\th) , \  \hyy_\th - \hxx_\th \right \rangle \dth \\
&\leq 0 + \sum_{i\in\I\esnu} \int_{\Theta_i} \norm{ \cc(\xxag^*)- v_\th( \xx^*_\th)}_2 \times \norm{ \hyy_\th - \hxx_\th }_2 \dth \\
&\leq \sum_{i\in\I\esnu} \int_{\Theta_i} (\Bc + \Bcut) \times \dset_\i \leq  (\Bc + \Bcut) \times  \mdset \ . \label{eq:proof_ut_maj1} 
\end{align}
To bound the first term of \eqref{eq:proof_profile_ut_two_ints_tobound}, we  divide it into two integral terms:
\begin{align}
&\int_{\Theta} \left\langle \cc(\hxxag)-v_\th(\hxx_\th)       , \ \hxx_\th - \xx^*_\th    \right \rangle \dth  \\ 
&=\sum_{\i\in\I\esnu}  \left[ \int_{\Theta_i} \left\langle \cc(\hxxag)-v_\i(\hxx_\i)       , \ \hxx_\th - \xx^*_\th    \right \rangle \dth  + \int_{\Theta_i} \left\langle v_\i(\hxx_\i)- v_\th(\hxx_\th)       , \ \hxx_\th - \xx^*_\th    \right \rangle \dth   \right] \ .
\end{align}
The first integral term is bounded using the characterization of a NE given in \Cref{prop:GNNash}:\begin{align}
&\sum_{\i\in\I\esnu}  \int_{\Theta_i} \left\langle \cc(\hxxag)-v_\i(\hxx_\i)       , \ \hxx_\th - \xx^*_\th    \right \rangle \dth \\
&= \sum_{\i\in\I\esnu}  \left\langle \cc(\hxxag)-v_\i(\hxx_\i)       , \ \hxx_\i - \xx^*_\i    \right \rangle  \\
& \leq \sum_{\i\in\I\esnu}  \left\langle \cc(\hxxag)-v_\i(\hxx_\i)       , \ \hxx_\i - \yy_\i^*   \right \rangle + \sum_{\i\in\I\esnu}  \left\langle \cc(\hxxag)-v_\i(\hxx_\i)       , \ \yy_\i^*- \xx^*_\i    \right \rangle \\
& \leq  -R(\hxx,\xx^*) +  \sum_{\i\in\I\esnu} \norm{\cc(\hxxag)-v_\i(\hxx_\i) }_2 \times \norm{ \yy_\i^*- \xx^*_\i  }_2 \\
& \leq 2C(M+1)^2 \mmu +  (\Bc+ \Bcut) \times 2 M \mdset \sum_{\i\in\I\esnu} \mu_i \\
&= 2C(M+1)^2 \mmu +  (\Bc+ \Bcut) \times 2 M \mdset \ . \label{eq:proof_ut_maj2}
\end{align}
For the second integral term, we use the distance between utilities \eqref{eq:def_dut}:
\begin{align}
&\sum_{\i\in\I\esnu} \int_{\Theta_i} \left\langle v_\i(\hxx_\i)- v_\th(\hxx_\th)       , \ \hxx_\th - \xx^*_\th    \right \rangle \dth   \\
& \leq  \sum_{\i\in\I\esnu} \mu_i \norm{v_\i(\hxx_\i)- v_\th(\hxx_\th)    }_2 \times \norm{ \hxx_\th - \xx^*_\th }_2 \\
& \leq \sum_{\i\in\I\esnu} \mu_\i \duti_\i \times 2\diamX \leq \mduti 2 \diamX \ . \label{eq:proof_ut_maj3}
\end{align}
We conclude the proof by combining \eqref{eq:proof_ut_maj1},\eqref{eq:proof_ut_maj2} and  \eqref{eq:proof_ut_maj3}.
\end{proof}
\fi

As in \Cref{lem:dist_profile_inside_part}, the  uniform strong concavity of the utility functions plays a key role in the convergence of disaggregated profiles $(\hxx_\th\esnu)_\snu$ to the nonatomic WE profile $\xx^*$.

\subsection{Construction of an Approximating Sequence}

In this section, we give examples of the construction of an  AAS for a nonatomic game $\Gna$, under two particular cases: the case of piecewise continuous functions and, next, the case of finite-dimensional parameters.

\subsubsection{Piecewise continuous parameters, uniform splitting}

\label{subsec:approx_uniform}
\newcommand{\disth}{\sigma} 
\newcommand{\cutp}{\upsilon} 
In this case, we assume that the parameters of the nonatomic game are piecewise continuous functions of $\th \in \Theta$: there exists a finite set of $K$ discontinuity points $0 \leq \disth_1 <\disth_2 < \dots < \disth_K\leq 1$, and the game is uniformly continuous (\Cref{def:continuity_nonatomic_game})  on $( \sigma_k, \sigma_{k+1})$, for each $k\in  \{0 ,\dots ,K+1\}$ with the convention $\sigma_0=0$ and $\sigma_K=1$.

For $\nu \in \nit^*$, consider the ordered set of $I_\snu$ cutting points $(\cutp_i\esnu)_{i=0}^{I_\snu} :=\left\{ \frac{k}{\nu} \right\}_{0\leq k \leq \nu} \cup \{\sigma_k   \}_{1\leq k \leq K}$ and define the partition $(\Theta_\i\esnu)_{i\in \I\esnu}$ of $\Theta$ by:
\vspace{-0.15cm}
\begin{equation}
\forall \i \in \{1, \dots, I_\snu\}\ , \Theta_i\esnu= [\cutp_{i-1}\esnu,\cutp_{i}\esnu ) .
\end{equation}

\begin{proposition}\label{prop:approx_seq_continuous}
For $\snu \in \nit^*$, consider the atomic game $\G\esnu$ defined with  $\I\esnu\eqDef\{ 1  \dots  I_\snu\}$, and for each $\i\in\I\esnu$:
\vspace{-0.15cm}
\begin{equation*}
%
\X_i\esnu \eqDef \mu_i\esnu \X_{ \bar{\cutp}_{i}\esnu} 
\text{ \ and \ } u_i\esnu\eqDef \xx \mapsto \mu_i\esnu u_{ \bar{\cutp}_{i}\esnu }\Big(\txt\frac{1}{\mu_i\esnu} \xx \Big)   ,
\vspace{-0.2cm}
\end{equation*}
with $\bar{\cutp}_{i}\esnu= \frac{\cutp_{i-1}\esnu +  \cutp_{i}\esnu}{2}$.  Then $\big(\G\esnu\big)_{\snu}=\big( \I\esnu, \T, \X\esnu, \cc, u\esnu \big)_{\snu} $ is an AAS of the nonatomic game  $\Gna=\left( \Theta, \T, \X_. , \cc, (u_\th)_{\th} \right) $. 
\end{proposition}

\begin{proof}
We have $I\esnu > \snu \longrightarrow \infty$ and for each $\i\in\I\esnu$, $\mu(\Theta_i\esnu) \leq \frac{1}{\nu} \longrightarrow 0$. The conditions on the feasibility sets and the utility functions are obtained with the piecewise uniform continuity  conditions. If we consider a common modulus of uniform continuity $\eta$ associated to an arbitrary $\varepsilon>0$, then, for $\nu$ large enough, we have, for each $\i \in <\I\esnu$, $\mu_\i\esnu< \eta $. Thus, for all $\th \in \Theta_\i\esnu$, $|\bar{\cutp}_{i}\esnu-\th|< \eta$, so that from the continuity conditions, we have:
\begin{align}
&  d_H\Big( \X_\th, \txt\frac{1}{\mu_\i\esnu}\X_\i\esnu\Big)=d_H( \X_\th, \X_{ \bar{\cutp}_{i}\esnu}) < \varepsilon \\
 \text{ and } & \max_{\xx\in\B_0(M)} \norm{ \nabla u\esnu_\i\left(\txt {\mu_i\esnu } \xx \right) - \nabla u_\th(\xx)}_2  = \norm{ \txt\frac{\mu_\i\esnu}{\mu_i\esnu}   \nabla u_{ \bar{\cutp}_{i}\esnu}\Big(\txt\frac{1}{\mu_i\esnu} \mu_\i\esnu \xx  \Big) -\nabla u_\th(\xx)  }_2 < \varepsilon \ , 
\end{align}
which concludes the proof.
\end{proof}

\subsubsection{Finite dimension, meshgrid approximation}
\label{subsec:approx_mesh}
\renewcommand{\ss}{\bm{s}}
\newcommand{\A}{\bm{A}}
\newcommand{\bb}{\bm{b}}
\newcommand{\us}{\underline{s}}
\newcommand{\os}{\overline{s}}
\newcommand{\dimp}{K} 
Consider a nonatomic routing game $\Gna=(\Theta,\X,\F)$ (\Cref{def:nonatomicGame}) satisfying the following two hypothesis:\begin{itemize}[wide]
\item The feasibility sets are $\dimp$-dimensional polytopes: there exist $\A \in \mathcal{M}_{\dimp,T}(\rit)$ and $\bb:\Theta\rightarrow\rit^\dimp$ bounded, such that for any $\th$, $\X_\th \eqDef \{\xx \in \rit^T \ ; \ \A\xx \leq \bb_\th \}$, with $\X_\th $ nonempty and bounded (as a polytope, $\X_\th$ is closed and convex).
\item  There exist a bounded function $\ss:\Theta \rightarrow \rit^q$ and a  function $u: \rit^q \times \B_0(M) \rightarrow \rit $ such that for any $\th \in \Theta, u_\th= u(\ss_\th, .)$.   Furthermore, $u$ is Lipschitz-continuous in $\ss$.
\end{itemize}

Let us define $\ub_k\eqDef  \min_\th b_{\th,k} $ and $\ob_k \eqDef \max_\th b_{\th,k}$ for $k \in \{1\dots \dimp \}$ and define similarly $\us_k,\os_k$ for $k\in \{1\dots q\}$.  

For $\nu \in \nit^*$, we consider the uniform meshgrid of $\nu^{\dimp+q}$ classes of $\prod_{k=1}^\dimp [\ub_k,\ob_k] \times \prod_{k=1}^q [\us_k,\os_k]$ which will give us a set of $I\esnu=\nu^{K+q}$ subsets. More explicitly, if we define:
\begin{equation*}
\Gamma\esnu=\{\nn=(n_k)_{k=1}^{\dimp+q} \in \nit^{\dimp+q} \,|\, n_k\in \{1,\ldots, \nu\}\}
\end{equation*}
the set of indices for the meshgrid, and with the cutting points $\ub_{k,n_k}\eqDef \ub_k+\frac{n_k}{\nu}(\ob_k-\ub_k) $ and $\us_{k,n_k}\eqDef \us_k+\frac{n_k}{\nu}(\os_k-\us_k)$ for $n_k\in \{0,\dots, \nu \}$,  we can define the subset $\Theta\esnu_{\nn}$ of $\Theta$ as:
\begin{align*}
\Theta\esnu_{\nn} \eqDef  & \Big\{\th\in \Theta \,|\forall  1\leq k \leq \dimp, b_{\th,k} \in [\ub_{k,n_k- 1},\ub_{k,n_k} [ \ ,  \forall 1\leq k \leq q ,  s_{\th, k}  \in [\us_{k,n_k-1},\us_{k,n_k}  [ \Big\}. 
\end{align*}
 
Since some of the subsets $\Theta\esnu_{\nn}$ can be of Lebesgue measure $0$, we define the set of players $\I\esnu$ as the elements $\nn$ of $\Gamma\esnu$ for which $\mu(\Theta\esnu_{\nn}) >0$. 

\begin{remark} \label{rm:approx_finitedim_infininitesimal}If there is a set of players of positive measure that have equal parameters $\bb$ and $\ss$, then the condition $\max_{\i\in\I\esnu} \mu_\i \rightarrow 0$ will not be satisfied. In that case, adding another dimension in the meshgrid by cutting $\Theta=[0,1]$ in $\nu$ uniform segments solves the problem.
\end{remark}

\begin{proposition}\label{prop:approx_seq_finitedim}
For $\snu \in \nit^*$, consider the atomic game $\G\esnu$ defined by:
\begin{align*}
&  \I\esnu\eqDef\{ \nn\in \Gamma\esnu | \mu(\Theta\esnu_{\nn})>0\} \ , \\
\text{and for each } \nn\in\I\esnu, \  
& \X_{\nn}\esnu \eqDef  \{\xx \in \rit^T  | \A\xx \leq \txt\int_{\Theta\esnu_{\nn}}\bb_{\th}\, \text{d}\th \} \ , \\
\text{ \ and \ } & u_{\nn}\esnu\eqDef  \xx \mapsto \mu_{\nn}\esnu u\left(\txt\frac{1}{\mu_{\nn}\esnu}\txt\int_{\Theta_{\nn}\esnu} \ss_{\th} \text{d}\th ,  \ \txt\frac{1}{\mu_{\nn}\esnu} \xx \right)  ,
\end{align*}
then the sequence of games $(\G\esnu)_{\snu}=\big( \I\esnu, \T, (\X\esnu_{\nn})_{\nn}, \cc, (u_{\nn}\esnu)_{\nn} \big)_{\snu} $
 is an AAS of the nonatomic game $\Gna=\left( \Theta, \T, (\X_\th)_\th , \cc, (u_\th)_{\th} \right) $. 
\end{proposition}

Before giving the proof of  \Cref{prop:approx_seq_finitedim}, we show the following nontrivial \Cref{lem:distance_perturb_polyhedra}, from which the convergence of the feasibility sets is easily derived.

%
%

\renewcommand{\L}{\Lambda}

\begin{lemma}
\label{lem:distance_perturb_polyhedra}
Given $\A \in \mathcal{M}_{\dimp,T}(\rit)$, define the parameterized polyhedra $\Lambda_{\bb} \eqDef\{\xx \in \rit^T \ ; \ \A\xx \leq \bb \}$ for $\bb$ in a bounded set $\mathcal{B}$. Assume that $\Lambda_{\bb}$ is nonempty for each $\bb\in\mathcal{B}$. Then the Hausdorff distance between polyhedra $\Lambda_b, \ b \in \mathcal{B}$,  is linearly bounded: there exists a constant $C_0>0$ such that:
\vspace{-0.15cm}
\begin{equation}
\forall \bb,\bb' \in \mathcal{B}, \  d_H(\L_{\bb}, \L_{\bb'}) \leq C_0 \|\bb-\bb'\| \ . 
\end{equation} 
\end{lemma}
\ifproofs
\begin{proof} The proof follows   \cite{batson1987combinatorial} in several parts, but we extend the result on the compact set $\B$, and drop the irredundancy assumption made in \cite{batson1987combinatorial}.
 
\newcommand{\V}{\mathcal{V}}
For each $\bb$, we denote by $V(\bb)$ the set of vertex of the polyhedron $\Lambda_{\bb}$. Under Assumption \ref{ass_X_nonat}, $V(\bb)$ is nonempty for any $\bb\in \B$.

First, as $\L_{\bb}$ is a polyhedra, we have $\L_{\bb}=\text{conv}(V(\bb))$ where $\text{conv}(X)$ is the convex hull of a set $X$. As the function $\xx \mapsto d(\xx,\Lambda_{\bb'})$ defined over $\L_{\bb}$ is continuous and convex, by the maximum principle, its maximum over the polyhedron $\L_{\bb}$ is achieved on $V(\bb)$. Thus, we have:
\begin{align}
d_H(\L_{\bb},\L_{\bb'}) &= \max [  \ \max_{\xx\in \L_{\bb}} d(\xx, \L_{\bb'}) \ , \max_{\xx\in \L_{\bb'}} d(\L_{\bb}, \xx) ] \\
 = &\max [  \ \max_{\xx\in V(\bb)} d(\xx, \L_{\bb'}) \ , \max_{\xx\in V(\bb')} d(\L_{\bb}, \xx) ] \\
 \leq  \max  & [  \ \max_{\xx\in V(\bb)} d(\xx, V({\bb'})) \ , \max_{\xx\in V(\bb')} d(V(\bb), \xx) ]  \\ 
 = & d_H\left( V(\bb),V(\bb') \right) \ . 
\end{align}
Let's denote by $ H_i(\bb)$ the hyperplane $\{\xx : A_i \xx = \bb_i\}$ and by $H_i^-(\bb)=\{\xx : A_i \xx \leq  \bb_i\}$ and $H_i^+(\bb)=\{\xx : A_i \xx \geq \bb_i\}$  the associated half-spaces. Then $\L_{\bb}=\bigcap_{i\in[1,m]} H_i^-(\bb)$.

\newcommand{\vv}{\bm{v}}
Now fix $\bb_0 \in \B$ and consider $\vv \in V(\bb_0)$. By definition, $\vv$ is the intersection of hyperplanes $\bigcap_{i\in K} H_i(\bb_0)$   where $K \subset \{1,\dots, m\}$ is maximal (note that $k \eqd \text{card}(K) \geq n$ otherwise $\vv$ can not be a vertex).

For $J\in \{ 1, \dots, m\}$, let $A_J$ denote the submatrix of $A$ obtained by considering the rows $A_j$ for $j\in J$. Let us introduce the sets of derived points (points of the arangement) of the set $K$, for each $\bb \in \mathcal{B}$: 
$$\V_K(\bb) \eqd \{ \xx \in \mathbb{R}^n \ ;  \exists J\subset K \ ; \ A_J \text{ is  invertible and }  \ \xx=A_J^{-1}\bb  \}   \ . $$
By definition, $\V_K(\bb_0)= \{ \vv \}$ and, for each $\bb\in \B$, $\V_K(\bb)$ is a set of at most $\binom{k}{n}$ elements.

First, note that for each $\bb \in \B $ and $\vv'\eqd A_J^{-1}\bb \in \V_K(\bb)$, one has: 
\begin{equation}
\label{eq:vertex_induced_bound}
\|\vv-\vv'\|= \| A_J^{-1}\bb_0- A_J^{-1} \bb\| \leq \|A_J^{-1}\| \times \| \bb_0 - \bb \| \leq \alpha \| \bb_0 - \bb \|
\end{equation}
where $\alpha \eqd \displaystyle\max_{A_J \text{ invertible} } \|  A_J^{-1}  \|$. 

Then, consider $\eta \eqd \min_{j \in \{1\dots m\}\setminus K} d\left(\vv,H_j^+(\bb)\right) $. By the maximality of $K$, $\eta>0$. As $\xx\mapsto  d(\xx,H_j^+)$ is continuous for each $j$, and from \eqref{eq:vertex_induced_bound} , there exists $\delta>0$ such that: $\|\bb_0 - \bb\| \leq \delta \ \Longrightarrow \forall \vv' \in \V_K(\bb), \min_{j \in \{1\dots  m\} \setminus K} d \left( \vv'  ,H_j^+(\bb)\right)> 0 $.

Next, we show that, for $\bb$ such that  $\| \bb_0 - \bb \| \leq \delta$, there exists $\vv'\in \V_K(\bb) \cap V(\bb)$. We proceed by induction on $k-n$. 
\renewcommand{\SS}{\mathcal{S}}

If $k=n$, then $\vv=A_K^{-1}\bb_0$ and for any $\bb $ in the ball $\SS_\delta(\bb_0)$, $\V_K(\bb)=\{A_K^{-1}\bb \}$.  Thus $\vv'\eqDef A_K^{-1}\bb$ verifies  $A_K \vv'=\bb_K$, and $A_j \vv' < b_j$ for all $j\notin K$, thus $\vv'$ belongs to  $V(\bb)$.

If $k=n+t$ with $t	\geq 1$, there exists  $j_0 \in K$ such that with $K'=K\setminus \{j_0\}$, $\V_{K'}(\bb_0)= \{\vv\}$. Consider the polyhedron $P=\bigcap_{i \in K'} H_i^-(\bb_0)$. By induction, there exists $J\subset K'$ such that $A_J^{-1}\bb$ is a vertex of $P$. If it satisfies also $A_{j_0}\xx \leq b_{j_0}$ then it is an element of $V(\bb)$. Else, consider a vertex $\vv'$ of the polyhedron $P \cap H_{j_0}^-(\bb)$ on the facet associated with $H_{j_0}(\bb)$. Then, $\vv'\in V_K(\bb)$ and, as $\bb\in \SS_\delta(\bb_0)$, it verifies $A_j \vv' < b_j $ for all  $j \notin K$, thus $\vv' \in V(\bb)$.

Thus, in any case and for $ \bb\in \SS_\delta(\bb_0), \ d(\vv,V(\bb)) \leq \|\vv-\vv'\| \leq \alpha \| \bb_0 - \bb \| $ and finally  $d \left(V(\bb_0), V(\bb) \right) \leq  \alpha \| \bb_0 - \bb \|$.
The collection $\left\{ \SS_{\delta_{\bb}(\bb) } \right\}_{\bb \in \B} $ is an open covering of the compact set $\B$, thus there exists a finite subcollection of cardinal $r$ that also covers $\B$, from which we deduce that there exists $D\leq \max(r\alpha)$ such that:
\begin{equation*}
\forall \bb, \bb' \in \B, \ d_H \left(V(\bb'), V(\bb) \right) \leq  D \| \bb' - \bb \| \ . \end{equation*}
\end{proof}
\else \fi

%

\begin{proof}[Proof of \Cref{prop:approx_seq_finitedim}]

First, to show the divergence of the number of players and their infinitesimal weight, we have  to follow \Cref{rm:approx_finitedim_infininitesimal} where we consider an additional splitting of the segment $\Theta=[0,1]$. In that case, we will have $I\esnu \geq \snu$ hence goes to positive infinity and for each $\nn \in\I\esnu, \mu(\Theta_{\nn}\esnu)\leq \frac{1}{\snu}$ hence goes to $0$.

Then, the convergence of the strategy sets follows from the fact that, for each $\nn\in\I\esnu$:
\begin{equation}
\txt\frac{1}{\mu_{\nn}\esnu}\X_{\nn}\esnu =\left\{\xx \in \rit^T  | \A\xx \leq \txt\frac{1}{\mu_{\nn}\esnu}\txt\int_{\Theta\esnu_{\nn}}\bb_{\th}\, \text{d}\th \right\} \ ,
\end{equation}
and from \Cref{lem:distance_perturb_polyhedra} which implies that, for each $\th'\in \Theta_{\nn}\esnu$:
\begin{align}
d_H\left( \X_{\th'}, \txt\frac{1}{\mu_{\nn}\esnu}\X_{\nn}\esnu \right) \leq C_0 \norm{ \bb_{\th'}-  \frac{1}{\mu_{\nn}\esnu}\int_{\Theta\esnu_{\nn}}\bb_{\th}\, \text{d}\th }
\leq \frac{C_0}{\nu} \norm{ \bm{\ob} -\bm{\ub} }_2 \ .
\end{align}

Finally, the convergence of utility functions comes from the Lipschitz continuity in $\ss$. For each  $\nn \in \I\esnu$ and  each $\th'\in \Theta_{\nn}\esnu$, we have:
\begin{align}
& \max_{\xx\in\B_0(M)} \norm{ \nabla u\esnu_{\nn}\left(\txt {\mu_{\nn}\esnu } \xx \right) - \nabla u_{\th'}(\xx)}_2  \\
 &= \max_{\xx\in\B_0(M)}  \norm{\nabla_{\xx} u\left(\txt\frac{1}{\mu_{\nn}\esnu}\txt\int_{\Theta_{\nn}\esnu} \ss_{\th} \text{d}\th ,  \ \txt\frac{1}{\mu_{\nn}\esnu} \xx \right) -\nabla_{\xx} u(\ss_{\th'},\xx) }_2  \\
&  \leq L_1 \norm{\frac{1}{\mu_{\nn}\esnu}\int_{\Theta_{\nn}\esnu} \ss_{\th} \text{d}\th - \ss_{\th'} }_2 \\
&\leq  \frac{L_1}{\snu} \norm{\bm{\os}-\bm{\us} }_2 \ , 
\end{align}
which terminates the proof.
\end{proof}

\begin{remark} Instead of using the average value on $\Theta_{\nn}\esnu $  in \Cref{prop:approx_seq_finitedim}, one could consider any value within the set $\Theta_{\nn}\esnu$. 
\end{remark}

Of course, the number of players considered in \Cref{prop:approx_seq_finitedim} is exponential in the dimensions of the parameters $K+q$, which can be large in practice. As a result, the number of players in the approximating atomic games considered can be very large, which will make the NE computation really long. Taking advantage of the continuity of the parametering functions  and following the approach of \Cref{prop:approx_seq_continuous} gives a smaller (in terms of number of players) approximating atomic instance.

\section{Numerical Application}
\label{sec:numericalExp}

We consider a population of consumers $\Theta=[0,1]$  with an energy demand  distribution $\th \mapsto E_\th$. Each consumer $\th$ splits her demand over $\T\eqDef\{O,P\}$, so that her feasibility set is $\X_\th\eqDef \{(\x_{\th,O},\x_{\th,P})\in\rit_+^2 \big| \x_{\th,O}+\x_{\th,P}=E_\th \}$. The index $O$ stands for offpeak-hours with a lower price $c_O(\xag)=\xag$ and $P$ are peak-hours with a higher price $c_P(\xag)=1+2\xag$.
The energy demand and the utility function in the nonatomic game are chosen as the piecewise continuous functions: 
 \vspace{-0.15cm}
\begin{equation*} 
E_\th\eqDef \left\{ \begin{array}{l}
  \sin(\pi \th)\, \text{ if }\th <0.7 \\
  0.3  \qquad\text{ if } \th \geq 0.7 
\end{array}
 \right. \ ,\ u_\th(\xx_\th)\eqDef \omega_\th \times \norm{ \yy_\th-\xx_\th}_2^2  \ ,
 \vspace{-0.15cm}
 \end{equation*}
 with $\yy_\th= (0,E_\th)$ the preference of user $\th$ for period $P$ and $\omega_\th\eqDef \th$ the preference  weight of player $\th$.
 
  We consider approximating atomic games by splitting $\Theta$ uniformly (\Cref{subsec:approx_uniform}) in 5, 20, 40 and 100 segments (players). We compute the NE for each atomic game using the best-response dynamics (each best-response is computed as a QP using algorithm  \cite{brucker1984n}, see \cite{PaulinArxivAnalysisImpl17} for convergence properties) and until the KKT optimality conditions for each player are satisfied up to an absolute error of $10^{-3}$.
\Cref{fig:cvgNahsWE}  shows, for each NE associated to the atomic games with 5, 20, 40 and 100 players, the linear interpolation of the load on the peak period $\x_{\th,P}$ (red filled area), while the load on the offpeak period can be observed as $\x_{\th,O}=E_\th-x_{\th,P}$.
\begin{figure}[ht!]
\centering
\vspace{-0.4cm}
\includegraphics[width=0.95\columnwidth]{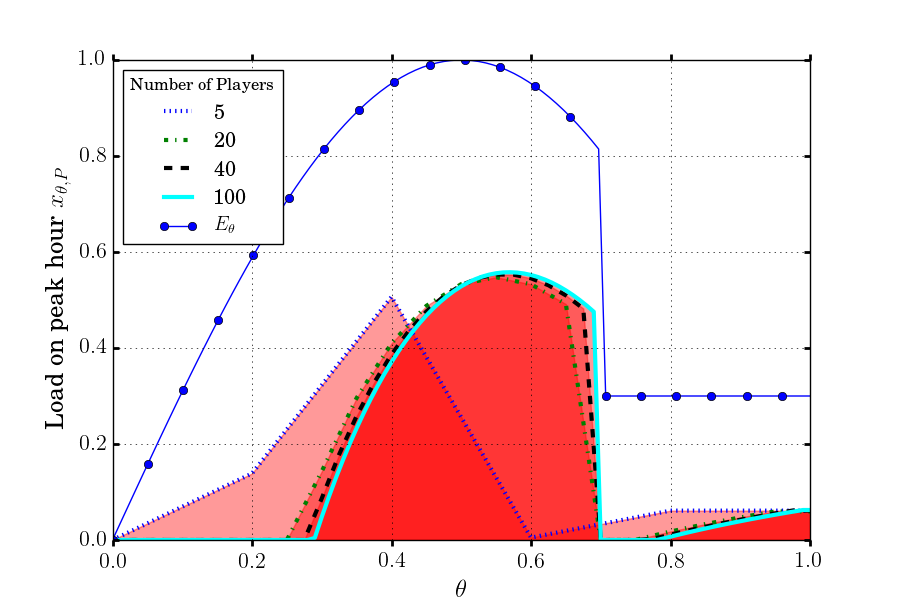}
\vspace{-0.3cm}
\caption{Convergence of the Nash Equilibrium profiles to a Wardrop Equilibrium profile}
\vspace{-0.3cm}
\label{fig:cvgNahsWE}
\end{figure}
We observe the convergence to the limit WE of the nonatomic game as stated in \Cref{thm:converge_with_u}. We also observe that the only discontinuity point  of $\th \mapsto \x^*_{\th,P}$ comes from the discontinuity of $\th\mapsto E_\th$ at $\th=0.7$, as stated in \Cref{prop:continuity_from_parameters}.

\section*{Conclusion}

This paper gives quantitative results on the convergence of Nash equilibria, associated to atomic games approximating a nonatomic routing game, to the Wardrop equilibrium of this nonatomic game. These results are obtained under different differentiability and monotonicity assumptions. Several directions can be explored to continue this work: first, we could analyze how the given theorems could be modified to apply in case of nonmonotone and nondifferentiable functions. Another natural extension would be to consider routing games on nonparallel networks or even general aggregative games: in that case, the separable costs structure is lost and the extension is therefore not trivial.

\section*{Acknowledgments}
\noindent We thank St\'ephane Gaubert, Marco Mazzola, Olivier Beaude and Nadia Oudjane for their insightful comments.


\bibliographystyle{ims}
\bibliography{../../bib/shortJournalNames,../../bib/biblio1,../../bib/biblio2,../../bib/biblio3,../../bib/biblioBooks}

\end{document}